\documentclass[a4paper,7pt,oneside,onecolumn,number,preprint,centertitle]{elsarticle}

\usepackage{amsmath,amssymb,bm}
\usepackage{graphicx}
\usepackage{textcomp}
\usepackage{xcolor}
\usepackage{amsthm}
\usepackage{tikz}

\usepackage{enumerate}
\usepackage{bbold}
\usepackage{epsfig}
\usepackage{multicol}
\usepackage{float}
\usepackage{hyperref}
\usepackage{multirow}
\usepackage{mathrsfs} %
\usepackage{dsfont}
\usepackage{algorithm}
\usepackage{algorithmicx}
\usepackage{algpseudocode}
\usepackage{flushend}

\usepackage{ifthen}
\DeclareMathAlphabet{\mbb}{U}{bbold}{m}{n} %
\newcommand \mb[1] {\ifthenelse{\equal{#1}{0}}{\mbb{0}}{\ifthenelse{\equal{#1}{1}}{\mbb{1}}{\mathbb{#1}}}} %

\renewcommand{\l}{\left}

\newtheorem{thm}{Theorem}

\newtheorem{cor}[thm]{Corollary}
\newtheorem{rem}{Remark}
\newtheorem{lem}[thm]{Lemma}
\newtheorem{ass}{Assumption}
\newtheorem{defn}{Definition}

\newtheorem{exmp}{Example}%

\renewcommand{\textcolor}[2]{#2}
\begin{document}
\begin{frontmatter}

\title{Scalable Design of Attack-Resilient Controllers for Positive Systems}

\author[inst1]{Alba Gurpegui}
    \ead{alba.gurpegui_ramon@control.lth.se}

\affiliation[inst1]{organization={Department of Automatic Control, Lund University},%
            addressline={Ole Römers väg 1}, 
            city={Lund},
            postcode={223 63}, 
            country={Sweden}}

\author[inst2,inst4]{Sribalaji C. Anand}
    \ead{srca@kth.se}

    \affiliation[inst2]{organization={Department of Electrical and Systems Engineering, University of Pennsylvania},%
            addressline={200 South 33rd Street}, 
            city={Philadelphia},
            postcode={PA 19104}, 
            country={United States}}

    \affiliation[inst4]{organization={Department of Decision and Control Systems, KTH},%
            addressline={Malvinas väg 10}, 
            city={Stockholm},
            postcode={10044}, 
            country={Sweden}}

\author[inst3]{Andr{\'e} M. H. Teixeira}
    \ead{andre.teixeira@it.uu.se}

    \affiliation[inst3]{organization={Department of Information Technology, Uppsala University},%
            addressline={Regementsvägen 10}, 
            city={Uppsala},
            postcode={751 05}, 
            country={Sweden}}

\begin{abstract}
This paper proposes a framework for secure and resilient controller design for positive systems against cyber-attacks. In particular, we consider a network-controlled system where an adversary injects false data into the actuator channels to increase the control cost (performance measure) while penalizing the attack effort and subject to state-dependent constraints. Using a minimax formulation, we analyze the worst-case performance loss caused by such adversaries, which is given by the solution of a difference equation, and an algebraic equation when the time horizon is infinite. We show that the optimal attack policy, among possible nonlinear policies, is linear. Despite the lack of explicit stealthiness constraints, we also show that when the measured output has an unstable zero which is not an unstable zero of the performance measure, the attacks can induce unbounded performance degradation. The proposed framework is also extended to systems with model uncertainty. Numerical examples illustrate the results and demonstrate how tools from positive systems and linear regulator theory can be used to mitigate cyber-attacks with low computational effort.

\footnotetext{
This work is supported by the Wallenberg AI, Autonomous Systems and Software Program (WASP), the Knut and Alice Wallenberg Foundation, the European Research Council under the  grant No 834142, by the Swedish Research Council under the grant 2021-06316, 2024-00185, and by the Swedish Foundation for Strategic Research.}
\end{abstract}


\end{frontmatter}

\section{Introduction}
\label{sec:introduction}
Network-controlled systems (NCSs) are an integral part of many critical infrastructures such as power grids \cite{wang2011wide}, water distribution systems \cite{amin2012cyber}, and production networks \cite{ahlen2019toward}. Due to the growing number of cyber-attacks on such NCS \cite{hemsley2018history}, its security has been a rich area of research for over two decades \cite{sandberg2022secure}. To build a resilient NCS, the first step is to quantify the worst case performance loss caused by stealthy attacks. For instance, the loss can be measured in terms of deviation caused in the state trajectory, state variance, control cost, etc (see \cite[section 3]{anand2025quantifying} for details). The second step is to design a controller that minimizes this worst-case performance loss. 

In general, the worst-case performance loss caused by an attack is computed by solving a semi-definite program (SDP) \cite{hashemi2022codesign}. However, such SDPs scale poorly for large-scale systems. 
Thus there has been a growing research interest to obtain scalable and resilient control design approaches \cite{milovsevic2017exploiting,nguyen2024scalable} which will be the focus of this paper. 

In particular, we consider a linear time-invariant (LTI) discrete time (DT) system controlled with a time-varying output feedback controller over the network. The controller obeys the actuator constraints which are enforced through the measured outputs. An adversary injects false data into the actuator channels to increase the control cost. Under this setup, we present the following contributions.
\begin{enumerate}
    \item \textcolor{blue}{We interpret the interaction between the adversary and the controller as a minimax optimal control problem and characterize the optimal strategies as a difference equation and an algebraic equation.}
    \item \textcolor{blue}{We derive conditions under which the presence of unstable plant zeros can lead to unbounded performance degradation, revealing structural vulnerabilities.}
    \item We prove that the optimal attack strategy is a (static) switching linear feedback policy in (in)finite horizon.
    %and a static linear feedback in the infinite horizon. 
    We derive the necessary conditions for an attack-resilient controller to exist.
    \item \textcolor{blue}{We extend the proposed framework to (a) cases when the adversary can have unbounded attack effort, and (b) parametric model uncertainties in the state matrix.} 
\end{enumerate}

While this paper builds on~\cite{albaEmmaAnders}, key differences include: an output feedback setting with a strategic adversary (vs. state feedback with bounded disturbances), and extensions to unbounded attack signals and model uncertainties; neither of which are addressed in \cite{albaEmmaAnders}. Finally, the scalable control of positive systems has been extensively studied in \cite{rantzer2015scalable} and references therein. Although scalable controller design for positive systems against disturbances, and in distributed settings is well understood, the resilient control design aspect against strategic adversaries remains poorly studied. This paper serves as a first step towards understanding the interplay between scalability and resilience in positive systems.

\textcolor{blue}{\textit{Notation:} Let $\mathbb{R}_{+}$ denote the set of nonnegative real numbers. The inequality $X > Y$ $(X \geq Y)$ mean that all the elements of the matrix $(X-Y)$ are positive (nonnegative).  A matrix $X$ is called positive if all the elements of $X$ are nonnegative but at least one element is nonzero. The notation $\left | X \right |$ means elementwise absolute value. For a vector $x$, $\mathrm{diag}(x)$ denotes the diagonal matrix with entries of $x$ on the diagonal. The signum of a scalar is defined as the set-valued map }
\vspace{-4pt}
$$\scalebox{1}{$\textcolor{blue}{\operatorname{sign}(x) =}$}
    \scalebox{1}{$\ensuremath{\textcolor{blue}{
    \begin{cases}
        \{-1\} & \text{if } x < 0\\
        [-1,+1] & \text{if } x = 0\\
        \{+1\} & \text{if } x > 0
    \end{cases}}}$}.$$
\section{Problem Formulation}\label{sec:prob:form}
In this section, we aim to describe the attack scenario, the adversarial policy, and finally state the problem studied in this paper. To this end, consider a DT LTI \textcolor{blue}{positive system:}
\begin{equation}\label{disc_sys_gral}
\begin{aligned}
    x[t+1]&=Ax[t]+B\tilde{u}[t],\; x[0] = x_0\\
    y[t]&=Cx[t] 
\end{aligned}    
\end{equation}
where $x[t] \in \mathbb{R}^n$ is the state of plant, $\tilde{u}[t] \in \mathbb{R}^m$ is the control input received by the plant, $y[t] \in \mathbb{R}^p$ is the plant output, and the matrices are of appropriate dimension.

\textcolor{blue}{An important feature of this framework is the positivity of the dynamics. Classical books on the topic are~\cite{BermanBook,Luenberger}. The positivity of the systems considered in this paper follows from the following result.}
\begin{lem}
   \textcolor{blue}{ Consider a DT LTI system $x[t+1]=Ax[t]$, $t\in [0,T]$ where $A\geq 0$. If $x[0]$ is entry-wise nonnegative, then $x[t]$ remains entry-wise nonnegative for all $t\in [0, T]$.}
\end{lem}
\begin{proof}
   \textcolor{blue}{ See~\cite[Theorem. 2.6]{Kaczorek}.}
\end{proof} 

Motivated by this property and similar to \cite{albaEmmaAnders}, we assume that the control inputs are bounded by the measured output according to $|u[t]|\leq E_y y[t]$ for all $t$, where $E_y$ is a nonnegative matrix of appropriate dimension. Since $y[t]=Cx[t]$, this constraint can be equivalently rewritten as $|u[t]|\leq Ex[t]$ where $E=E_yC$. We present further discussions on $E_y$ when discussing the adversary model.

We next consider an adversary that injects false data into the actuator channels for a given time horizon $T$, i.e.
\begin{align}
    \tilde u[t]=u[t]+B_a a[t], \hspace{2mm} \forall t \in \left\{1,\dots, T \right\}
\end{align}
where $a[t] \in \mathbb R^l$ is the false data injected by the adversary, and $B_a\in \mathbb R^{m \times l}$. Then the dynamics in \eqref{disc_sys_gral} under attack becomes
\begin{equation}\label{sys}
    \begin{aligned}
    x[t+1] &= Ax[t] + B{u}[t] + Fa[t], x[0] = x_0\\
    y[t] &= Cx[t]
\end{aligned}
\end{equation}
where $F \triangleq BB_a$. Next we provide a detailed description of the attack policy and constraints. 
\subsection{Attack policy}
Similar to \cite{wu2018optimal}, the adversary injects false data as a 
function of the system state
\begin{align*}
    a[t]=\mu_a(C_ax[t]), \;\mu_a: \mathbb{R}^n \to \mathbb{R}^l
\end{align*}
where $C_a\in \mathbb R^{p_a \times n}$ maps the system states to the adversary's observed states $y_a[t] = C_a x[t]$, and $\mu_a$ 
is a function possibly dependent on the system matrices.

In the literature, linear attack policies of the form $a[\textcolor{blue}{t}] = \Delta C_a x[\textcolor{blue}{t}]$ are commonly considered \cite{guo2016optimal,guo2018worst}, where $\Delta$ maps the observed states to the attack signal. Since $\Delta$ is not known beforehand, we characterize the class of all such attacks by the constraint
\begin{equation}\label{eq:attack_constraint}
    \vert a[t] \vert \leq G_a C_a x[t],
\end{equation}
where $G_a\in \mathbb R^{l \times p_a}$ is a shaping matrix chosen so that $G=G_a C_a$ is nonnegative.
For chosen $G_a$ the constraint \eqref{eq:attack_constraint} captures all linear attack policies when $\|\Delta\|_{\infty} \leq 1$.
Note that the constraint \eqref{eq:attack_constraint} also encompasses nonlinear and time-varying policies, making it strictly more general than those in the literature. The matrix $E_y$ plays a similar role to $G_a$ but for the controller policies. Thus, the set of admissible attack policies is given by
\begin{align}\label{cons_set_unbound}
\mathcal A_{Gx}&:=\left\{  a\in \mathbb R^l : |a|\leq G x, \hspace{0.7mm}\text{with}\; G \in \mathbb R^{l \times n}_+\right\}.
\end{align}
\begin{rem}\label{rem:G:Ca}
If 
% $G$ is chosen such that 
$\mathrm{ker}(G)=\mathrm{ker}(C_a)$, the bound $|a|\leq Gx$ depends only on the adversary's observed states. In particular, if $C_ax[t]=0$, then $Gx[t]=0$, so admissible attacks vanish along the adversary's unobservable direction. Moreover, $G_a$ allows the constraint to focus on specific state components, so that $G$ reflects both the attacker's information and the states targeted by the adversary. $\hfill \triangleleft$
\end{rem}
\subsection{Problem definition}
Combining the arguments above, we consider an adversary that injects false data by solving the following optimization problem under Assumption~\ref{ass_prob_form}
\begin{equation}\label{prob_form}
    \begin{aligned}
    \textcolor{blue}{\min_{|u|\leq E x} \max_{a\in \mathcal A_{Gx}}}  & \quad \sum_{t = 0}^{T} \left( s^{\top}x[t]+ r^{\top}u[t]-\alpha^{\top}a[t] \right ) \\
    \text{s.t.} & \quad u[t]=\mu(x[t]),\; a[t]=\mu_a(x[t])\; \\
    & \quad \text{System dynamics \eqref{sys}}
    \end{aligned}
\end{equation}
where $\mathcal A_{Gx}$ represents the constraint set defined in \eqref{cons_set_unbound}, $E=E_yC$, $s \in \mathbb R^n_+$, $r \in \mathbb R^m$, and $\alpha \in \mathbb R^l$. Here, $s$ is a vector with positive, nonzero elements at the positions corresponding to the target state. Thus, the term $s^{\top}x[t]$ represents the benefit (state deviation) the adversary gains from manipulating the state, and $\alpha$ penalizes the magnitude of the injected attack. We next establish the following.
\begin{ass}\label{ass_prob_form}
\textcolor{blue}{$A\geq|B|E +|F|G, \hspace{2mm} s \geq E^{\top}|r|-G^{\top}|\alpha|.$$\hfill \triangleleft$}
\end{ass}
\textcolor{blue}{The first condition in Assumption \ref{ass_prob_form} ensures invariance of the positive orthant under the system dynamics. In fact this assumption implies that both the open-loop and closed-loop system dynamics correspond to those of a positive system. The second condition guarantees that the cost functional in~\eqref{prob_form} is nonnegative and bounded from below.}
\begin{rem}\label{rem:E:C}
 Recall that $E=E_yC$. If $E$ is chosen such that $\ker(E)=\ker(C)$, then the input constraint $|u|\leq Ex$ depends only on the measured output. In particular, any trajectory satisfying $Cx[t]=0$ also satisfies $Ex[t]=0$, meaning that the admissible control collapses along unobservable directions. This is relevant when stealthy attack policies are exploited with respect to the measured output (e.g. zero-dynamics attacks). In the sequel, we call an attack stealthy when the measured output stays near zero despite the presence of a non-zero attack signal. $\hfill \triangleleft$
\end{rem}
\section{Main Results}\label{sec:bounded}
In this section, we provide a solution to the design problem \eqref{prob_form} in both finite and infinite horizon. The results of this section is organized as follows. Firstly in Theorem~\ref{thm_sup_boundG_fin} we discuss the results for finite horizon, and in Corollary~\ref{cor_sup_boundG_inf} we present the results for infinite horizon. In Theorem \ref{thm:zeros} we show that the optimal cost diverges in the presence of an invariant zero. Finally, we depict the results via numerical examples. 
%\vspace{-10pt}
\subsection{Solution to \eqref{prob_form} in the finite and infinite horizon}
In this section, we present a solution to the optimal control problem \eqref{prob_form}. Before presenting the result, we briefly recall the definition of the matrices $E = E_yC$ and $F=BB_a$. 
\begin{thm}[Finite horizon]\label{thm_sup_boundG_fin}
Consider the system \eqref{sys}, and suppose that \textcolor{blue}{Assumption~\ref{ass_prob_form} holds.} Then for any finite value $T \in \mathbb{Z}^{+}$, the following hold
\begin{enumerate}
    \item The optimal control problem \eqref{prob_form} has the optimal value $p_0^{\top}x_0$ with $p_T=0$ and
    \begin{equation}\label{seq_sup_fin_bound}
p_{t}=s+A^{\top}p_{t+1} + \begin{bmatrix}
    -E^{\top} & G^{\top}
\end{bmatrix}\begin{bmatrix}
|r+B^{\top}p_{t+1}|\\
|F^{\top}p_{t+1}-\alpha|
\end{bmatrix}
\end{equation}
\item The optimal attack policy, among all potentially nonlinear attack policies, is given by a linear time-varying feedback gain $a[t]=L^*[t]x[t]$ where
    \begin{align}\label{L_star_fin}
        L^*[t]\in \text{diag}(\text{sign}(-\alpha+F^{\top}p_t))G. \;\; \square
    \end{align}
\end{enumerate}
\end{thm}
\begin{rem}
    The optimal control policy is also a switching policy~\cite{LR_paperII} and is given by $u[t]= -K^*[t]x[t]$ with
    \begin{align}\label{K_t}
        K^*[t] \hspace{0.5mm}\textcolor{blue}{\in} \hspace{0.5mm} \text{diag}(\text{sign}(r+p^{\top}_t B))E. \;\; \triangleleft
    \end{align}
\end{rem}
The proof of Theorem~\ref{thm_sup_boundG_fin} and all the other results in the sequel can be found in the appendix. The contributions of Theorem~\ref{thm_sup_boundG_fin} are twofold. \textcolor{blue}{Firstly, the theorem provides a framework for the operator to jointly derive the \emph{optimal} control strategy against a worst-case adversary, and to characterize the optimal attack strategy that such an adversary would deploy.} Unlike prior studies that addressed only attack design \cite{wu2018optimal,degue2022stealthy}, this approach extends the formulation to a co-design setting. 
Secondly, many works in the literature assume a linear attack policy and design the optimal gain. However, our work is the first to prove that the linear attack policy is the optimal policy among all policies. Next, we extend the result in Theorem~\ref{thm_sup_boundG_fin} to infinite horizon. 

\begin{cor}[Infinite Horizon]\label{cor_sup_boundG_inf}
\textcolor{blue}{Suppose Assumption~\ref{ass_prob_form} holds and the algebraic equation }
\begin{equation}\label{alg_eq_min_max_bound}
\textcolor{blue}{ p=s+A^{\top}p-E^{\top}|r+B^{\top}p|+G^{\top}|-\alpha + F^{\top}p|}
\end{equation}
\textcolor{blue}{admits a nonnegative solution $p\in \mathbb R^n_+$. Then by~\cite[Theorem~1]{albaEmmaAnders}, $p$ is the limit of the recursion~\eqref{seq_sup_fin_bound}, the optimal value of~\eqref{prob_form} is $p^{\top}x_0$, and the optimal attack policy is $a^*[t]=L^*x[t]$, with $L^* \in \mathrm{diag}(\mathrm{sign}(-\alpha+F^{\top}p))G$.} $\hfill \square$
\end{cor}

Corollary \ref{cor_sup_boundG_inf} states that the optimal value of \eqref{prob_form} can be found by solving an algebraic equation in the infinite horizon. \textcolor{blue}{Most works in the literature solve \eqref{prob_form} via an SDP \cite{anand2025quantifying} where the number of paramaters to solve grow quadratically in the state dimension. However, our work is the first to propose an algebraic equation instead, making the problem scalable. This scalability is due to the fact that the number of parameters in \eqref{alg_eq_min_max_bound} grow linearly in the state dimension.}
\begin{rem}
The right hand sides of~\eqref{L_star_fin} \textcolor{blue}{and~\eqref{K_t}} are sets, since multiple attacks \textcolor{blue}{and control actions} may exist that achieve the optimal cost. In particular, for any index $i$ such that $-\alpha_i+[F^{\top}p_t]_i=0$, it holds that $L^{*} \in \mathcal{L}^{*}$, where
%all feedback gain matrices are in the set
    \begin{align*}
        \textcolor{blue}{\scalebox{1}{$\mathcal{L}^*=\left\{DG \hspace{1mm} | \hspace{1mm} D_{ii} \in [-1,1], \; D_{jj}\in \text{sign}(-\alpha_j + [p_t^{\top}F]_j) \hspace{1mm}\mathrm{ for } \hspace{1mm} j \neq i\right\}$} }
    \end{align*}
\textcolor{blue}{Analogously, for any index $i$ such that $r_i+p(t)_i^{\top}B_i=0$ it holds that $K^*\in \mathcal K^*$ where}
\begin{align*}
   \textcolor{blue}{ \scalebox{1}{$\mathcal{K}^*=\left\{DE \hspace{1mm} | \hspace{1mm} D_{ii}\in \left[ -1, 1\right], D_{jj} \in \operatorname{sign}(r_j+[p_t^{\top} B]_j) \hspace{1mm}\mathrm{ for } \hspace{1mm} j\neq i \right\}$}.} \triangleleft
\end{align*} 
\end{rem}
\subsection{Zero-dynamics attack}
In this section, we discuss the geometric growth of the optimal cost in the presence of unstable zeros. Before we present the results, we introduce some definitions.
\begin{defn}[Unstable invariant zero]
Given a tuple $\Sigma = (A,B,C)$, $\lambda \in \mathbb{C}$ is an invariant zero of $\Sigma$ if $\exists\;x_0 \in \mathbb{R}^n, g \in \mathbb{R}^m$ such that 
\begin{equation}\label{eq:zero}
    \begin{bmatrix}
        \lambda I-A & B\\
        C & 0
    \end{bmatrix}\begin{bmatrix}
        z_0\\ g
    \end{bmatrix} = \begin{bmatrix}
        0 \\ 0
    \end{bmatrix}.
\end{equation}
The zero is defined to be unstable if $ \vert \lambda \vert > 1$, and $(\lambda,z_0,g)$ is defined as the zero tuple. $\hfill \triangleleft$ 
\end{defn}

For a system with an unstable invariant zero, there exists a non-zero input which makes the states grow in magnitude whilst the outputs are close to zero. Such inputs are called zero-dynamics. 
Equivalently, using \cite[Theorem~1]{teixeira2012revealing}, the zero dynamics can also be characterized using geometric control.
\begin{defn}[Zero dynamics using geometric control]\label{def:geometric}
The input $a[t]=Lz[t]$ with $z[t+1] = (A+BL)z[t]$, $(A+FL)\mathcal V \subseteq \mathcal V \subseteq \text{ker} \hspace{1mm}C$ and $z_0 = z[0] \in \text{ker} \hspace{1mm}C $ yields $y[t]=0, \forall t \geq 0$ with $x[0]=z_0$.
$\hfill \triangleleft$ 
\end{defn}

Next, we show that under specific conditions, the adversary can exploit the zero dynamics to construct stealthy attacks while driving the cost unbounded. 
\begin{thm}\label{thm:zeros}
    Consider the optimal control problem~\eqref{prob_form}. Let $(\lambda,x_0,g)$ be an unstable zero tuple associated with $\Sigma=(A,F,E)$ which satisfies $x_0 \geq 0$ and $\vert g \vert \leq Gx_0$. Suppose there exists a matrix $L$ such that $Lx_0=g$, where $x_0\in \text{ker}\hspace{1mm} C$, $(A+FL)\mathcal V \subseteq \mathcal{V} \subseteq \text{ker}\hspace{1mm} C$ and $|Lx|\leq Gx, \; \forall x \in \mathcal V$. Consider the feedback attack policy $a[t]=Lx[t]$ with $x[0]=x_0$. Then the optimal cost of \eqref{prob_form} is given by
\begin{equation}\label{eq:optimal:zero:cost}
(s-L^{\top}\alpha)^{\top}x_0\textstyle\sum\nolimits_{i=0}^{T-1}\lambda^i
\end{equation}  
if $(s-L^{\top}\alpha)^{\top}>0$ and $(s-L^{\top}\alpha)^{\top}x_0>0$. $\hfill \square$
\end{thm}

In short, Theorem~\ref{thm:zeros} states that if there exists an unstable zero, and the value of $\alpha$ is small enough, the optimal attack policy is the \emph{stealthy} zero dynamics attack. Additionally, since $\lambda>1$, the optimal cost in \eqref{eq:optimal:zero:cost} diverges as $T\to\infty$. The feedback zero-dynamics attack $a[t]=Lx[t]$ is therefore admissible (satisfies the input constraints), stealthy, and produces unbounded performance degradation. Next we illustrate the results through a numerical example. 
\begin{rem}
The zero-dynamics attack requires that $x_0\in \mathcal V$ to obtain identically zero output. When $x_0 \notin \mathcal V$, a transient appears at the output due to the initial condition mismatch. When $\lambda>1$, the unstable zero dynamics dominates the transient dynamics, and eventually governs the system behavior \cite{teixeira2012revealing}. $\hfill \triangleleft$.
\end{rem}
\begin{exmp}
To illustrate the previous results, we consider a DT LTI system of the form \eqref{sys} with
\begin{align*}
&        \scalebox{1}{$A=\begin{bmatrix}
            0.95&0.08&0.045\\
            0.2&0.67&0.05\\
            0&0.02&0.4
        \end{bmatrix}, \hspace{1mm}
        B=\begin{bmatrix}
            0.12&0\\
            0&0.1\\
            0.01&-0.01
        \end{bmatrix},\hspace{1mm} C= e_3^\top$}\\
& \scalebox{1}{$B_a = \begin{bmatrix}
            0&1
        \end{bmatrix}^{\top},\; F =\begin{bmatrix}
            0&0.1&-0.01
        \end{bmatrix}^{\top}, \; E_y = \begin{bmatrix}
        0.05 & 
        0.05
    \end{bmatrix}^\top $}
\end{align*}
and $G=  \begin{bmatrix}0.1&2&0\end{bmatrix}$.
By computing the transfer function from the attack input a to the output y, we verify that the system admits a real unstable transmission zero at $\lambda \approx 1.04$. 

Let us set $x_0=\begin{bmatrix} 0.3602 & 0.4172 & 0 \end{bmatrix}^\top $, $g=0.8344$.
This choice ensures that $x_0$ lies along the system's internal zero-dynamics mode corresponding to the unstable zero $\lambda= 1.04$ while also satisfying $Ex_0=0$ and $|g|\leq Gx_0$. 

Next, choose the state and control weights in the cost function of the problem~\eqref{prob_form} as $s=2 \cdot \mathds 1$, $r=\mathds 1$ and $\alpha=1$. Using the backwards recursion we compute the nominal cost-to-go sequence $p_t$ and the corresponding sequence for the system under attack. Fig.~\ref{exmp1_states} represents the nominal state trajectories, the states in open-loop under zero dynamics attack $a[t]=\lambda^t g$ and in closed-loop with nominal feedback $u[t]=-Kx[t]$ and feedback attack $a[t]=Gx[t]$. Fig.~\ref{exmp_1_cost} depicts the cost-to-go $\mathds 1^{\top}p_t$ for the nominal system and the system under attack. Here it can be observed that the nominal cost (blue) grows approximately linearly with the horizon, whereas the cost under zero-dynamics attack (orange) increases geometrically at a rate determined by the unstable zero $\lambda$. $\hfill \triangleleft$ 
\end{exmp}
\begin{figure*}[t]
    \centering
    \includegraphics[width=1\textwidth]{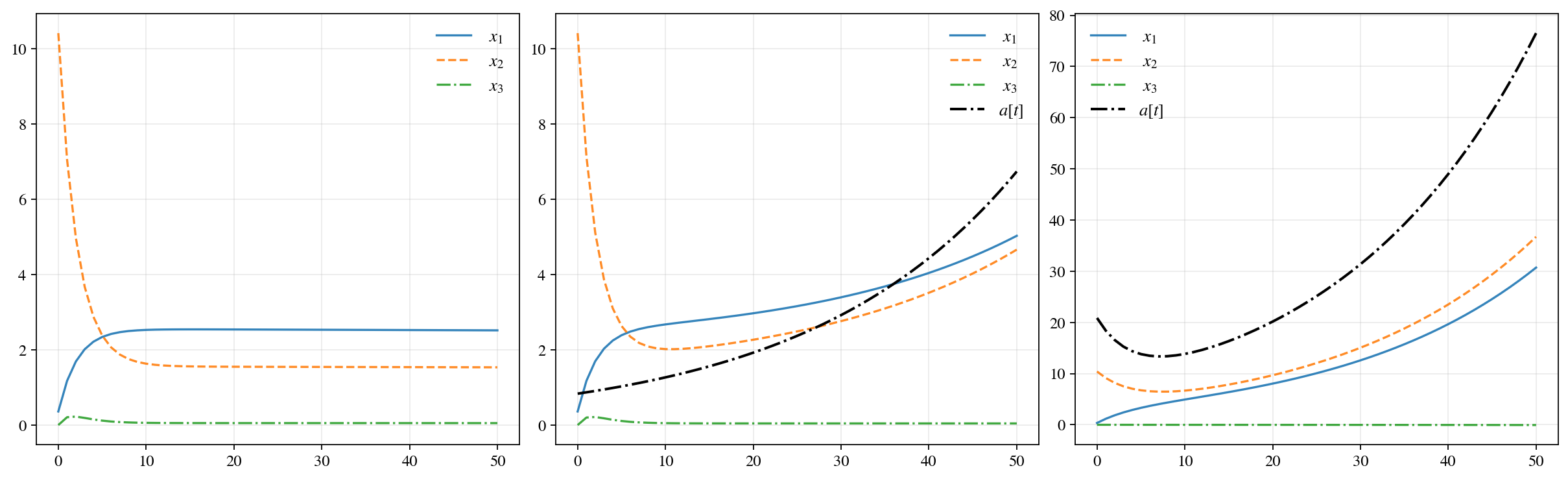}
    \caption{Nominal state trajectories (left), open-loop system under attack (middle), and closed-loop system under attack (right). The zero-dynamics attack drives the system state along the unstable zero direction, resulting in unbounded growth in the optimal cost while remaining stealthy.}
    \label{exmp1_states}
\end{figure*}
\begin{figure}[t]
        \centering
        \includegraphics[width=6cm]{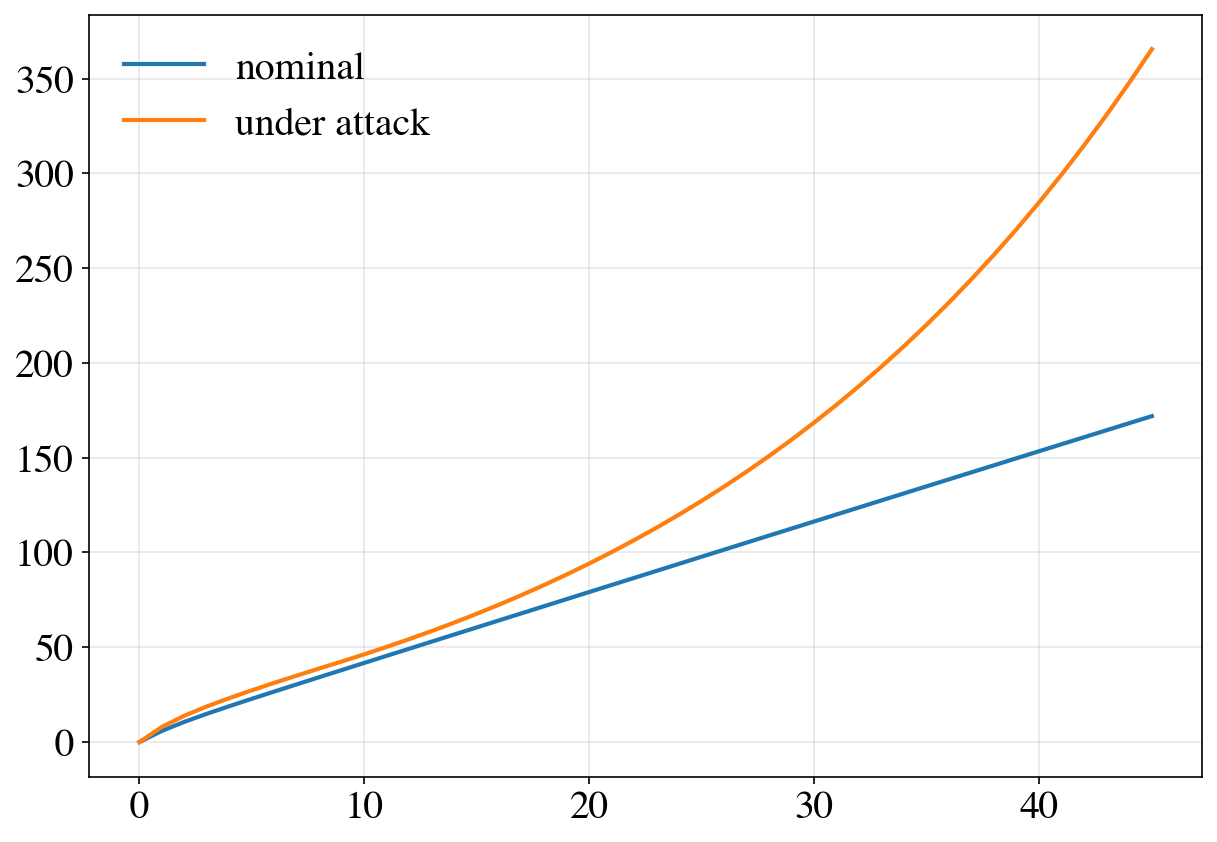}
        \caption{Nominal and Attacked $\mathds 1^{\top} p_{\textcolor{blue}{t}}$ for $\Sigma$.}
        \vspace{-15pt}
        \label{exmp_1_cost}
    \end{figure}
\section{Additional Discussions}\label{sec:unbounded}
In this section, we extend the main results in two directions. Firstly, we consider an adversary that does not adhere to the constraint in \eqref{eq:attack_constraint}. Secondly, we extend the results to uncertain systems where we consider an operator with access to only the nominal model.
\subsection{Unconstrained attack signals}
In this section, we consider attack signals lie in the set
\begin{align*}
    \mathcal A_{\geq 0}&:=\left\{ a\in \mathbb R^l  : a \geq 0\right\}.
\end{align*}
Then we aim to solve \eqref{prob_form} where the constraint $a \in \mathcal A_{Gx}$ is replaced by $a\in \mathcal A_{\geq 0}$. Let us call this optimization problem $(P2)$. The following result provides necessary and sufficient conditions under which $(P2)$ has a finite solution.
\begin{thm}\label{thm_sup_unb_fin}
Consider the system \eqref{sys} and suppose that \textcolor{blue}{Assumption~\ref{ass_prob_form} holds.} Then 
the optimal control problem $(P2)$ has the optimal value $p_0^{\top}x_0$ with 
    \begin{equation}\label{seq_sup_fin_unbound}
        p_{t}=s+A^{\top}p_{t+1}-E^{\top}|r+B^{\top}p_{t+1}|,\;p_T=0
    \end{equation} 
    {if and only if $\alpha \geq F^{\top}p_0$.} $\hfill \square$
\end{thm}

The proof is similar to the proof of Theorem~\ref{thm_sup_boundG_fin} and is thus omitted. Observe that Theorem~\ref{thm_sup_unb_fin} establishes that $(P2)$ admits a finite value if and only if $\alpha \geq F^{\top}p_0.$ This condition ensures that the penalty on the attack effort outweighs the adversarial gain, keeping the attack bounded. If for some time $\bar t$ we have $\alpha < F^{\top}p_{\bar t}$, this balance is lost and both the attack signal and the cost grow unbounded as $t \rightarrow \bar t$. The term $\alpha - F^{\top}p_{\bar{t}}$ defines the \textit{boundedness margin}, a nonnegative margin keeps the attacker inactive $(a[t]=0)$, while a negative margin leads to unbounded adversarial action. Next we illustrate the results presented in this section through a numerical example. 

\begin{exmp}
 Consider a simple water distribution process consisting of two storage tanks supplied by three pumps. The state vector is the water level in the tanks, and the dynamics are of the form \eqref{sys} where 
    \begin{align*}
        A = \begin{bmatrix}
            0.92&0.03\\
            0.15&0.06
        \end{bmatrix}, \hspace{1mm} B=\begin{bmatrix}
            1&0&0.4\\
            0&1&0.7
        \end{bmatrix}, \hspace{2mm} C=\begin{bmatrix}
            1&0
        \end{bmatrix}.
    \end{align*}
Suppose that an adversary corrupts pumps $1$ and $3$, i.e.:
    \begin{align*}
        B_a=\begin{bmatrix}
            1&0&0\\
            0&0&1
        \end{bmatrix}^{\top}, \hspace{1mm} F=\begin{bmatrix}
            1&0.4\\
            0&0.7
        \end{bmatrix}.
    \end{align*}
 The adversary gains utility from making both tank levels large, so we set $s=2 \cdot \mathds{1}$, $r=0$  and compute the backwards recursion~\eqref{seq_sup_fin_unbound}. Then we synthesize the corresponding controller gain $K =\text{diag}(\mathrm{sign}(r+B^{\top}p))E_yC$ where $E_y=\begin{bmatrix}
    0.02&0.02&0.02
\end{bmatrix}^{\top}$.

Denote the boundedness margin for attack channel $i$ at time $t$ as $m_t=\min_{i=1,2}[\alpha_i-(F^{\top}p_t)_i]$. By Theorem~\ref{thm_sup_unb_fin}, when the value of $m_t$ becomes positive at a given time instant ${t}^*$, the adversary can choose an attack that drives the cost unbounded at time ${t}^*$ or later. To illustrate the impact of the weight $\alpha$ on the boundedness margin, we compare three different choices of $\alpha$. The top panel in Fig.~\ref{examp_alpha} shows for $x_0=\begin{bmatrix}
    1&0.5
\end{bmatrix}^{\top}$ how the same optimal controller can be counteracted by the attack driving the system unstable at different times $t^*(\alpha)$. $\hfill \triangleleft$
The table below lists the three stealthiness weights $\alpha$ considered in the experiment and their corresponding first-violation times $t^*$. We can infer from this example that a larger value of $\alpha$ helps in mitigating powerful adversaries which does not consider stealthiness constraints. 
\begin{figure}[h]
        \centering
        \includegraphics[width=7cm]{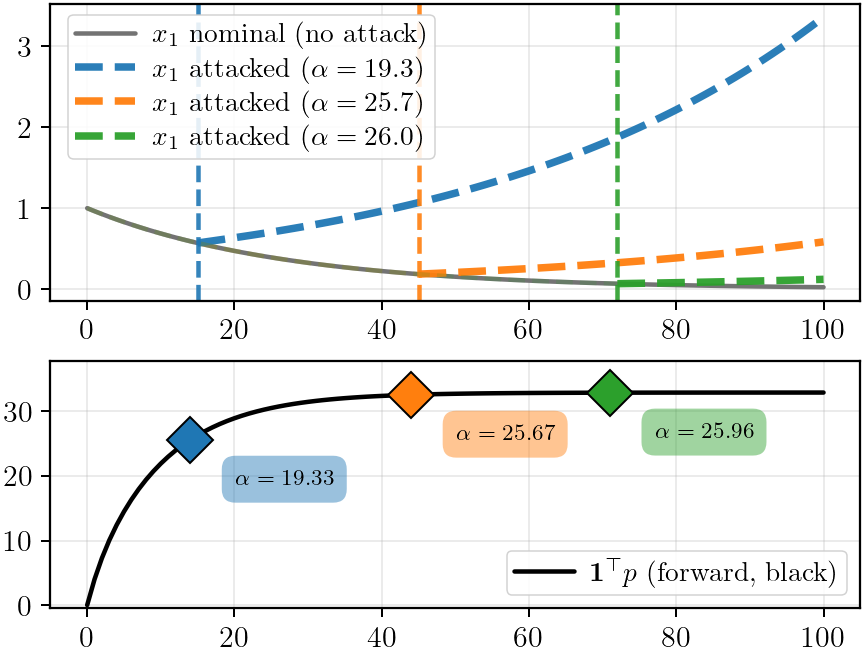}
        \caption{Top: Closed-loop state $x_1[t]$ under persistent attack with starting times $t^*(\alpha)$. Bottom: forward iteration of $\mathds 1^{\top}p_t$ starting at $\mathds 1^{\top}p_0=0$ where colored diamonds represent the last safe index before the boundedness margin becomes  negative for each $\alpha$.}
        \vspace{-15pt}
        \label{examp_alpha}
    \end{figure}
\end{exmp}
\subsection{Uncertain systems}
Consider a dynamical system $\Sigma = (A,B,C)$ with states denoted by $x_{\Sigma}$, and whose corresponding attack-optimal controller $K_{\Sigma}$ is obtained by solving~\eqref{seq_sup_fin_bound}.  During deployment, the actual process dynamics may deviate from the nominal model $\Sigma$. This deviation can result from modeling errors or slowly time-varying dynamics. We denote the actual process dynamics as $\Sigma_r = (A_r,B,C)$, with states denoted by $x_{\Sigma_r}$. We now pose the following question: Suppose the controller $K_{\Sigma}$ designed for the nominal system $\Sigma$ is deployed to control the actual system $\Sigma_r$. Is the controller $K_{\Sigma}$ admissible for $\Sigma_r$? In this section we provide initial insight into this question under additional assumptions. 

\begin{lem}\label{thm_uncertain}
Suppose that $A$ and $A_r$ both satisfy \textcolor{blue}{Assumption~\ref{ass_prob_form} holds}, $r\geq 0$, $F \geq 0$ and $B\geq 0$. Consider the optimal controller for the nominal system $K_{\Sigma}$ in \eqref{K_t}. Then $K_{\Sigma}$ is also admissible for $\Sigma_r$. $\hfill \square$
\end{lem}

Lemma~\ref{thm_uncertain} states that when $A$ and $A_r$ ensures invarinace of the positive orthant, with positive costs ($r \geq 0$), then under some very mild assumptions $F \geq 0$ and $B\geq 0$, the nominal controller is admissible in the presence of uncertainty. We next illustrate the above results through an example.
\begin{exmp}
Let $s=\begin{bmatrix}
            0.6 &
            0.8 &
            0.2
        \end{bmatrix}^\top$, $r=\mathds 1_2$, 
    \begin{align*}
        &\scalebox{1}{$A= \begin{bmatrix}
            0.33&0.33&0.22\\
            0.22&0.11&0.11\\
            0.55&0.66&0.55
        \end{bmatrix}, \hspace{1mm} A_r=\begin{bmatrix}
            0.42&0.28&0.14\\
            0.28&0.14&0.14\\
            0.84&0.98&0.84
        \end{bmatrix}$}\\
        &\scalebox{1}{ $B=\begin{bmatrix}
            0.3&0.1\\
            0&0\\
            0.4&0.5
        \end{bmatrix}, \hspace{1mm} C= \begin{bmatrix}
            0.24&0.36\\
            0.28&0.32\\
            0.2&0
        \end{bmatrix}^{\top}, \hspace{1mm} B_a=\begin{bmatrix}
            0.5&0.3\\
            0.2&0.2
        \end{bmatrix}$}\\ 
        &\scalebox{1}{$E_y=\begin{bmatrix}
            0.6&0\\
            0.48&0.12
        \end{bmatrix},
        \hspace{1mm} G=\begin{bmatrix}
            0&0.1&0.4\\
            0.3&0.3&0.2
        \end{bmatrix}, \;\alpha= 3 \mathds 1_2.$}
    \end{align*}
First we observe that the optimal controller for both configurations is static $u[t]=E x[t]$. Thus the controller is admissible for $\Sigma$ and $\Sigma_r$. However, as expected, the optimal attack policy changes signs differently over time for $A$ and $A_r$. This can be verified from the backward recursion for $p_t$, as the attack sign at time $t$ depends on $F^{\top}p_t-\alpha$. By inspection over a time horizon $T=50$, for the nominal matrix $A$, $F^{\top}p_t-\alpha<0$ for all $t=0,\dots T$, whereas for $A_r$ the sign pattern switches, with channel 2 switching at $t=42$ in and channel 1 at $t=46$.
The maximum model mismatch of $|A-A_r|\approx 0.32$ results in an increase of the open-loop spectral radius $|\rho(A)-\rho(A_r)|\approx 0.26$, which results in different closed-loop system performance. Analyzing this performance difference is left for future work. $\hfill \triangleleft$
\end{exmp}
\section{Conclusions}\label{sec:con}
We addressed resilient controller design against stealthy cyber-attacks using a minimax formulation. We proved that the optimal attack and control policies are linear, and can be obtained by solving a difference equation, and an algebraic equation in the infinite-horizon case. Conditions under which the performance cost becomes unbounded in terms of unstable system zeros were presented. Future work includes designing an uncertainty-aware resilient controller.
%\bibliographystyle{unsrt} %plaindin
%\bibliography{bibliography}

\bibliographystyle{abbrv} 
\bibliography{cas-refs}

\appendix

\section*{Proof of Theorem~\ref{thm_sup_boundG_fin}}
As it is shown in~\cite{albaEmmaAnders}, for the minimax optimal control problem class~\eqref{prob_form}  with linear dynamics, linear stage costs, and homogeneous state-dependent constraints admit linear value functions. Consequently, the Bellman recursion yields a cost-to-go of the form $V_t(x)=p_t^{\top}x$ where $p_t$ evolves according to the backwards recursion~\eqref{seq_sup_fin_bound}, which proves claim 1). Using the same dynamic programming argument, the optimal attack policy $a[t]=\mu_a(x[t])$ is obtained by solving
\begin{align*}
    &\scalebox{1}{$\mu_a(x) = \mathrm{argmax}_{ | a|\leq Gx}
    \big [ s^{\top}x+r^{\top}u-\alpha^{\top}a + p^{\top}_t(Ax+Bu+Fa) \big ]$}\\
   &~~~~~~~~\scalebox{1}{$=\mathrm{argmax}_{ | a|\leq Gx} \textstyle\sum\nolimits_{i=1}^{m}\left[ \left ( -\alpha_{i}+p^{\top}_t F_{i} \right )a_{i} \right].$}
\end{align*}
Since for each $i = 1, \hspace{1mm}... \hspace{1mm} l$ the constraint $\left|a \right | \leq Gx$ restricts $a_{i}$ to the interval $\left [ -G_{i}x, G_{i}x \right ]$, the maximization is attained at the boundary of this interval. Hence, the optimal choice satisfies $a_{i} \in \mathrm{sign}(-\alpha_{i}+p^{\top}F_{i})G_{i}x,\hspace{1mm}i = 1 \hspace{1mm}... \hspace{1mm} l$ 
which yields the optimal attack policy in feedback form. $\hfill \blacksquare$

\section*{Proof of Theorem~\ref{thm:zeros}}
Before we present the proof, we present some preliminary results which helps us with the proof. 
First, in Lemma~\ref{lem_zero_dyn_open}, we show that when the system has a real zero, it is possible to construct an admissible attack sequence that produces state trajectories whose growth rate is determined by the zero’s eigenvalue.
\begin{lem}\label{lem_zero_dyn_open} 
Suppose there exists a unstable zero tuple $(\lambda, x_0,g)$ associated with $\Sigma=(A,F,E)$ which satisfies $\lambda >0$, $x_0\geq 0$ and $|g|\leq Gx_0$. Let $x[0]=x_0$ and define open-loop attack policy $a[t]=\lambda^t g, u[t]=0$. Then, the trajectories of \eqref{sys} in open loop becomes $x[t]=\lambda^t x_0, \hspace{2mm} \text{for all} \hspace{2mm} t=0,\dots, T,$ the attack is admissible and the total cost of \eqref{prob_form} under attack is
\begin{equation}\label{eq:cost}
    J(x_0)=(s^{\top}x_0-\alpha^{\top}g)\textstyle\sum\nolimits_{i=0}^{T-1}\lambda^i.
\end{equation}
\end{lem}
\begin{proof}[Proof of Lemma~\ref{lem_zero_dyn_open}]
   Consider the system dynamics \eqref{sys} under attack with $u[t]=0$. 
   % \begin{align*}
   %     x[t+1]=Ax[t]+Fa[t]. 
   % \end{align*}
   By assumption, $(\lambda, x_0, g)$ is an unstable zero tuple of $(A,F,E)$ which implies $Ax_0+Fg=\lambda x_0.$
   % \begin{align*}
   %     Ax_0+Fg=\lambda x_0.
   % \end{align*}
   Define the open-loop attack policy $a[t]=\lambda^t g$ and initialize the system at $x_0 \geq 0$.
   % $x[t]=\lambda^t x_0$. 
   It can be shown by induction that the resulting state trajectory satisfies $x[t]=\lambda^t x_0$ for all $t=0, \dots, T$. By definition this holds for $t=0$. Suppose it holds for some $t\geq 0$. Then $x[t+1]=\lambda^t(Ax_0+Fg)=\lambda^t(\lambda x_0)=\lambda^{t+1}x_0$, which completes the induction. Next, we verify the admissibility of the attack. 
   Since  $x_0\geq 0$, and $\lambda >0$, it follows that $x[t]=\lambda^t x_0 \geq 0$ for all $t$. Moreover, the assumption $|g|\leq Gx_0$ implies that $|a[t]|=\lambda^t |g|\leq \lambda^t Gx_0=Gx[t]$ so the contraint $|a[t]|\leq Gx[t]$ is satisfied at all times. 
   Finally, the stage cost at time $t$ reduces to $s^{\top}x[t]-\alpha^{\top}a[t]=\lambda^t(s^{\top}x_0-\alpha^{\top}g)$. Summing the stage cost over $t=0,\dots, T-1$  yields~\eqref{eq:cost} which completes the proof.
\end{proof}

Next, in Lemma~\ref{lem_zero_dyn_cl}, we show that the open-loop zero-dynamics behavior can be replicated through a feedback law of the form $a[t] = L x[t]$. This feedback realization enables a consistent representation of the zero-dynamics attack within the dynamic-programming (DP) framework of the paper. Recall that the recursions in Theorem~\ref{thm_sup_boundG_fin} arise form applying DP recursions of the form:
\begin{align*}
p_{t}^\top x[t]  &= s^\top x[t]  \\
&~~+ \min_{|u|<Ex[t]} \max_{|a|<Gx[t]}\left( r^\top u - \alpha^\top a + p_{t+1}^\top x[t+1]\right) 
\end{align*}

\begin{lem}\label{lem_zero_dyn_cl}
Let $(\lambda, x_0, g)$ be the unstable zero tuple associate with $(A,F,E)$. Suppose there exists a subspace $\mathcal{V} \subseteq \ker{C}$ and a matrix $L$ such that $(A+FL)\mathcal{V} \subseteq \mathcal{V}$, $x_0 \in \mathcal{V}$ and $Lx_0=g$. Assume that $|Lx|\leq Gx$ for all $x\in \mathcal{V}$. Consider the attack policy $a[t]=Lx[t]$ with $x[0]=x_0$. Then, for trajectories restricted to $\mathcal{V}$, the DP recursion of \eqref{prob_form} satisfies
    \begin{equation}\label{eq:zero:DP}
         p_t=\lambda p_{t+1}+s-L^{\top}\alpha, \hspace{2mm}p_T=0.
    \end{equation}
    In particular,
    \begin{equation}\label{eq:zero:DP2}
        p_t=(s-L^{\top}\alpha)\sum_{k=t}^{T-1}\lambda^{T-k-1}, \hspace{0.1mm} p_0=(s-L^{\top}\alpha)\sum_{i=0}^{T-1}\lambda^i.
    \end{equation}
\end{lem}
\begin{proof}[Proof of Lemma~\ref{lem_zero_dyn_cl}]
    Lemma~\ref{lem_zero_dyn_open} shows that the open-loop policy $a[t]=\lambda^t g$ with $x[0]=x_0$ generates the zero-dynamics trajectory $x[t]=\lambda^t x_0$ and produces a cost that grows geometrically with~$\lambda$.
    Lemma~\ref{lem_zero_dyn_cl} establishes that, under the stated assumptions on $L$ and $\mathcal V$, the feedback realization $a[t]=Lx[t]$ reproduces the same state evolution within~$\mathcal V$, and that the corresponding cost recursion satisfies \eqref{eq:zero:DP}. Solving backward with $p_T=0$ yields \eqref{eq:zero:DP2}.
\end{proof}

Lastly, we next provide the proof of Theorem~\ref{thm:zeros}.
 
\begin{proof}[Proof of Theorem~\ref{thm:zeros}]
Lemma~\ref{lem_zero_dyn_open} shows that the open-loop zero-dynamics input $a[t]=\lambda^t g$ generates the trajectory $x[t]=\lambda^t x_0$ and yields a cost growing as $\sum_{i=0}^{T-1}\lambda^i$.
Lemma~\ref{lem_zero_dyn_cl} shows that, under the stated conditions on $L$ and $\mathcal V$, the feedback law $a[t]=Lx[t]$ reproduces the same trajectory within $\mathcal V$ and that the cost recursion reduces to \eqref{eq:zero:DP}. Solving backwards gives \eqref{eq:zero:DP2}, and at $t=0$ yields \eqref{eq:optimal:zero:cost}. 

By assumption $(s-L^{\top}\alpha)^{\top}x_0 >0$. Hence, when $\lambda>1$ the geometric sum grows unbounded with $T$. Moreover, since $\mathcal V \subseteq \text{ker} \;C$, the resulting trajectories satisfy $Cx[t]=0$ for all $t$, and the attack remains stealthy with respect to the measured output. Since $E=E_yC$, it follows that $Ex[t]=0$ along these trajectories, which enforces $u[t]=0$ through the admissibility constraint. Consequently, the controller cannot mitigate the zero-dynamics attack, and the cost grows unbounded while the attack remains stealthy. 
\end{proof}

\section*{Proof of Lemma~\ref{thm_uncertain}}

Suppose $p_{t+1} \geq 0$. Since $r \geq 0$, and $B \geq 0$, we have $|r + B^\top p_{t+1}| = r + B^\top p_{t+1}$. Since $F \geq 0$, we have $F^\top p_{t+1} \geq 0$, and by the reverse triangle inequality, $|F^\top p_{t+1} - \alpha| \geq F^\top p_{t+1} - |\alpha| = |F^\top|p_{t+1} - |\alpha|$. Substituting into \eqref{seq_sup_fin_bound} gives
\begin{align*}
p_t &\geq s - E^\top r - G^\top|\alpha| + (A^\top - E^\top B^\top + G^\top|F^\top|)p_{t+1}.
\end{align*}
The first term is nonnegative by $s \geq E^\top |r| - G^\top|\alpha|$, and the second is nonnegative since $A \geq |B|E + |F|G$ implies $A^\top - E^\top B^\top + G^\top|F^\top| \geq 2G^\top|F^\top| \geq 0$ and $p_{t+1} \geq 0$. Hence $p_t \geq 0$, completing the induction.

Since $p_t\geq 0$, it follows that $r+B^{\top}p_t\geq 0$, and the optimal control law~\eqref{K_t} reduces to $u[t]=Ex[t]$, independent of the uncertain state matrix. Hence $K_{\Sigma}$ is admissible for both $\Sigma$ and $\Sigma_r$, which concludes the proof. $\hfill \blacksquare$

\end{document}